\documentclass{article}
\usepackage{amsthm,framed}
\usepackage[sumlimits]{amsmath}
\usepackage{amsfonts}
\usepackage{color,url}

\usepackage{graphicx}  
\DeclareMathOperator{\ad}{ad}
\DeclareMathOperator{\curl}{curl}

\DeclareMathOperator{\diff}{d}

\usepackage[margin=2cm]{geometry}
\newtheorem{theorem}{Theorem}

\newtheorem{corollary}[theorem]{Corollary}
\newtheorem{definition}[theorem]{Definition}
\newtheorem{example}[theorem]{Example}
\newtheorem{proposition}[theorem]{Proposition}
\newtheorem{remark}[theorem]{Remark}
\def\MM#1{\boldsymbol{#1}}
\newcommand{\pp}[2]{\frac{\partial #1}{\partial #2}} 
\newcommand{\dede}[2]{\frac{\delta #1}{\delta #2}}
\newcommand{\dd}[2]{\frac{\diff #1}{\diff #2}}

\newcommand{\bfi}[1]{{\bfseries\itshape #1}}
\newcommand{\rem}[1]{}

\def\contract{\makebox[1.2em][c]{\mbox{\rule{.6em}
{.01truein}\rule{.01truein}{.6em}}}}

\begin{document}
\title{On Noether's theorem for the Euler-Poincar\'e equation on the
  diffeomorphism group with advected quantities} 
\author{
Colin J. Cotter$^{1}$ and Darryl D. Holm$^{2}$
}
\addtocounter{footnote}{1}
\footnotetext{Aeronautics Department, Imperial College, London SW7 2AZ, UK.  
\texttt{colin.cotter@ic.ac.uk}
\addtocounter{footnote}{1} }
\footnotetext{Department of Mathematics, Imperial College, London SW7 2AZ, UK. 
\texttt{d.holm@ic.ac.uk}
\addtocounter{footnote}{1}}
  \date{
}
\maketitle

\centerline{\it In honor of Peter Olver's 60-th birthday}

\begin{abstract}

  We show how Noether conservation laws can be obtained from the
  particle relabelling symmetries in the Euler-Poincar\'e theory of
  ideal fluids with advected quantities. All calculations can be
  performed without Lagrangian variables, by using the Eulerian vector
  fields that generate the symmetries, and we identify the
  time-evolution equation that these vector fields satisfy.  When
  advected quantities (such as advected scalars or densities) are
  present, there is an additional constraint that the vector fields
  must leave the advected quantities invariant. We show that if this
  constraint is satisfied initially then it will be satisfied for all
  times. We then show how to solve these constraint equations in 
  various examples to obtain evolution equations from the conservation
  laws.  We also discuss some fluid conservation laws in the
  Euler-Poincar\'e theory that do not arise from Noether symmetries,
  and explain the relationship between the conservation laws obtained
  here, and the Kelvin-Noether theorem given in Section 4 of Holm,
  Marsden and Ratiu, {\it Adv. in Math.}, 1998.

\end{abstract}

\tableofcontents

\section{Introduction} \label{sec-intro}

As Noether did in her famous paper \cite{Noether1918}, we are dealing
with invariant variational principles. This subject has a vast
literature and has been a favorite topic for Peter Olver, to which he
returned many times \cite{BeOl1982,Ol1984a,Ol1984b,Ol1984c,Ol1986,Olver-book,MuRoOl2006}.

A Lie group transformation that leaves the Lagrangian invariant in 
Hamilton's principle is called a \emph{variational Lie symmetry}. The correspondence between variational Lie symmetries and conservation laws for Euler-Lagrange equations is completely determined by the  Noether's First Theorem \cite{Ko-Sc2004,Noether1918,Olver-book}. Namely, every variational Lie symmetry yields a conservation law.%
\footnote{Noether's celebrated paper \cite{Noether1918} contains two major theorems. The present paper discusses only the first of these theorems. For good discussions of the second Noether theorem, see e.g. \cite{Noether1918,Olver-book,BaCaJa1994,Ko-Sc2004,HyMa2011}.}
Our main goal here is to identify explicitly in terms of \emph{Eulerian observables} the vector fields of the relabelling symmetry transformations under the Lie group $G$ of smooth invertible maps that are responsible for some of the well-known conservation laws in the Euler-Poincar\'e theory of fluids with advected quantities \cite{HoMaRa1998}. In particular, we treat a few hands-on examples in fluid dynamics that recover some famous formulas such as helicity of fluids, Ertel's potential vorticity in geophysical fluid dynamics
(GFD) and Chandrasekhar's cross-helicity for magnetohydrodynamics
(MHD).  We also discuss the relation of the classical Noether's Theorem with the Kelvin-Noether circulation theorem from the Euler-Poincar\'e theory of ideal fluids with advected quantities in \cite{HoMaRa1998}.  In addition, we discuss conservation laws in the Euler-Poincar\'e theory that do not arise from Noether symmetries.
Finally, we discuss some applications of Noether's Theorem in image registration problems. 

It seems that every theoretical physicist and many mathematicians eventually feel compelled to write a paper about Noether's Theorem. Previous influential papers along similar lines about Noether's Theorem in fluid dynamics related to the directions taken here include \cite{So1976,De1978,Si1985,Ho1986,AbHo1987,PaMo1996a,PaMo1996b,Br2010,HyMa2011} and of course references therein.
\\

The main content of the paper is:
\begin{enumerate}
\item Section \ref{sec-formulateEP} briefly summarises the
  Euler-Poincar\'e formulation of ideal fluid dynamics with advected
  quantities. In particular, we summarise several simple but useful
  theorems that are available for studying how the Noether theorem
  associates variational Lie symmetries with conservation laws for fluids.
\item Section \ref{sec-examples} uses these theorems in a sequence of
  examples that derive several of the most well-known conservation
  laws for ideal fluids in the Euler-Poincar\'e formulation \cite{HoMaRa1998}.  
\item
Section \ref{nonNoetherCLs} points out that not all fluid conservation laws follow from Noether's theorem, by considering the counterexample of magnetic helicity for MHD. It also makes a connection between Noether's theorem as discussed in this paper, and the Kelvin-Noether circulation theorem discussed in \cite{HoMaRa1998}.
\item 
  Section \ref{sec-conclude} discusses some numerical issues and applications of these ideas outside of fluid dynamics. Section \ref{sec-conclude} also raises topics for future research inspired by Lie symmetries and Noether's theorem.
\end{enumerate}

\section{Formulation} \label{sec-formulateEP}

We begin by 
laying out the assumptions that underlie the Euler-Poincar\'e formulation. These are the following.
\begin{enumerate}
\item There is a right representation of the action of a Lie group
  $G$ on its tangent space $TG$ and on the vector space $\mathcal{V}$. The
  action on $TG\times \mathcal{V}$ is denoted by concatenation on the right, as
  $(v_g , a) h= (v_gh, ah)$ for $g,h\in G$.
\item The Lagrangian $L: TG \times \mathcal{V}\to \mathbb{R}$ is right
  $G$-invariant.
\item In particular, if $a_0 \in \mathcal{V}$, define the Lagrangian $L_{a_0} :
  TG\to \mathbb{R}$ by $L_{a_0}(v_g)=L(v_g , a_0)$. Then $L_{a_0}$ is
  right-invariant under the lift to $TG$ of the right action of
  $G_{a_0}$ on $G$, where $G_{a_0}$ is the isotropy group of $a_0$.
\item Right $G$-invariance of the Lagrangian $L$ permits us to define
  a \emph{reduced Lagrangian} $l : \mathfrak{g}\times \mathcal{V}\to
  \mathbb{R}$ by
\begin{equation}\label{red-lag}
l(v_g g^{-1}, a_0 g^{-1})=L(v_g , a_0).
\end{equation}
Conversely, this relation defines for any 
$l : \mathfrak{g}\times \mathcal{V}\to\mathbb{R}$ a right $G$-invariant
function $L: TG \times \mathcal{V}\to \mathbb{R}$.
\item
For a curve $g(t) \in G$, let $u(t) :=\dot{g}(t) g(t)^{-1}$ and define the curve $a(t)\in{\mathcal{V}}$ obtained from the action $G\times \mathcal{V}\to \mathcal{V}$ as the unique solution of the linear differential equation with time dependent coefficients 
\begin{equation}\label{advect1}
\left(\partial_t + \mathcal{L}_{u(t)}\right)a(t) = 0
\,,
\end{equation}
with initial condition $a(0)=a_0$ and Lie derivative $\mathcal{L}_{u(t)}$.
\end{enumerate}

For fluids, the Lie group $G={\rm Diff}(\mathbb{R}^3)$ is the group of
diffeomorphisms of three-dimensional space. This is the Lie group of
smooth invertible maps defined on $\mathbb{R}^3$ and with smooth
inverses.%
\footnote{
Strictly speaking, $G={\rm Diff}(\mathbb{R}^3)$ denotes the connected component at the identity
of the diffeomorphisms.
Its Lie algebra comprises the right-invariant vector fields on $\mathbb{R}^3$, denoted as $\mathfrak{X}(\mathbb{R}^3)$.}
At time $t$, the curve $g(t)$ defines the mapping from a reference
configuration (known as \emph{label space}) to the physical domain so
that $x(t)=g(t)x_0$ describes Lagrangian particle trajectories for
each label $x_0$. 

\begin{definition}\rm
The solution $a(t)=a_0 g(t)^{-1}$ 
of equation (\ref{advect1}) is called an
\emph{advected quantity} for fluids, and the right-invariant vector field
$u(t) :=\dot{g}g(t)^{-1}\in\mathfrak{X}(\mathbb{R}^3)$ is called the \emph{Eulerian, or
spatial, fluid velocity}. 
\end{definition}

\begin{remark}\rm
Examples of advected quantities include the extensive thermodynamic properties that are carried by fluid elements such as their heat and mass. Equation (\ref{advect1}) means physically that along the flow $g(t)$ of the vector field $u(t)$ the fluid elements are to be regarded as closed thermodynamic systems that do not 
exchange heat and mass with their neighbours.

Some particular examples of advected quantities that we discuss in this paper are:
\begin{enumerate}
\item Scalar fields (0-forms) $a(t)=s$ that satisfy:
\[
\left(\partial_t + \mathcal{L}_{u(t)}\right)s
= \left(\partial_t + \MM{u}\cdot\nabla\right)s=0.
\]
In geophysical models scalar advected quantities includes buoyancy due to heat and salinity.
\item Density fields (volume forms) $a(t)=\rho dV$ that satisfy
\[
\left(\partial_t + \mathcal{L}_{u(t)}\right)\rho dV
= \left(\partial_t\rho + \nabla\cdot(\MM{u}\rho)\right)\diff{V}=0.
\]
This type of advected quantity is used for the fluid density, or layer
depth in shallow water models.
\item Flux fields (2-forms) $a(t)=\MM{B}\cdot\diff{\MM{S}}$ that
  satisfy
\[
\left(\partial_t + \mathcal{L}_{u(t)}\right)\MM{B}\cdot\diff{\MM{S}}
= \left(\partial_t\MM{B} -
{\rm curl}(\MM{u}\times\MM{B}) 
\right)
\cdot\diff{\MM{S}}=0.
\]
This type of advected quantity is used, e.g., for the magnetic flux in
magnetohydrodynamics.
\end{enumerate}
For more discussion of advected quantities, see \cite{HoMaRa1998}.
The back-to-labels map that specifies the label of the fluid parcel currently at a given spatial position would also be an advected quantity. However, in this paper, we shall restrict ourselves to dealing only with \emph{Eulerian observables}, and the particle label is not observable at any given Eulerian point in a fluid flow.
\end{remark}

\subsection{Euler-Poincar\'e theorem with advected quantities}

Here we review the approach presented in \cite{HoMaRa1998} to
obtaining the variational equation of motion, known as the
Euler-Poincar\'e equation, for general reduced Lagrangians $l(u,a)$
with advected quantities.

Hamilton's principle 
$\delta S=0$ for $S=\int l(u,a)\,dt$
with the reduced Lagrangian
defined in equation (\ref{red-lag}) may be expressed either abstractly as
\begin{equation}
\label{action1}
0  
=
\delta S
 = \delta \int_{t_0}^{t_1} l(u,a)\diff{t} 
=   \int_{t_0}^{t_1} \left\langle \dede{l}{u} \,,\,\delta u \right\rangle 
+  \left\langle \dede{l}{a}\,,\,\delta a \right\rangle \diff{t}
\,,\end{equation}
where angle brackets denote appropriate pairings, or equivalently in coordinates 
with Lagrangian $\int_\mathcal{D} \ell({u},a)\diff{V}$
\begin{equation}
\label{action2}
0  
=
\delta S
 =  \delta \int_{t_0}^{t_1} \int_\mathcal{D} \ell({u},a)\diff{V}\diff{t} 
  =  \int_{t_0}^{t_1} \int_\mathcal{D} \left(\dede{\ell}{{u}}\cdot\delta {u} + 
\dede{\ell}{a}\,\delta a\right) \diff{V}\diff{t},
\end{equation}
where $\mathcal{D}$ is the spatial domain with boundary $\partial\mathcal{D}$
on which the fluid velocity has no normal component; that is, ${u\cdot n}=0$. The expressions 
$\dede{l}{u}\in\mathfrak{X}^*$ and $\dede{l}{a}\in \mathcal{V}^*$ are variational
derivatives in $u$ and $a$, respectively.  
We ensure that variations in $g$ honour the boundary conditions, by defining $\delta g =
w\circ g$, in which $w$ is a vector field whose components satisfy ${w\cdot n}=0$ on $\partial\mathcal{D}$.  

In the remainder of the paper, we will find it convenient to use a hybrid notation that passes freely between the abstract notation and the more explicit coordinate notation, as in equations (\ref{action1}) and (\ref{action2}). We believe this hybrid notation, whose meaning will always be clear from the context,  will appeal to a wider readership than the abstract  notation. Conversely, we will sometimes find that the calculations we need to perform are written more directly in the abstract notation using the language of differential forms. 

The infinitesimal transformations for $u$ and $a$ are \cite{HoMaRa1998}
\begin{equation}
\delta{u} 
= {\dot w}-{\rm ad}_u w
:= {\dot w}+[u,w]\,, \quad
\delta a = -\mathcal{L}_{w}a\,.
\end{equation}
Here $a$ denotes any quantity that is advected with the flow,
\emph{e.g.} scalar tracers ${s}$, densities $\rho\diff{V}$ \emph{etc.} The
linear operator on $w$, ${\rm ad}_u$, is defined in terms of
$[u,{w}]$, which is the commutator (Lie bracket) of the vector fields
$u$ and $w$ in $\mathfrak{X}(\mathbb{R}^3)$. 
Furthermore, we seek the stationary point $\delta S=0$ in Hamilton's principle above, subject to $\delta g = 0$ at the endpoints $t=t_0$ and $t=t_1$; hence, we also require $w=\delta g\, g^{-1}$ to vanish at the endpoints.

Substitution in (\ref{action1}) now yields
\begin{equation}\label{HP1}
0 = - \int_{t_0}^{t_1}\!\! \int_{\mathcal{D}} \left(
\pp{}{t}\dede{l}{{u}}  + \ad_u^*\dede{l}{u}
 - \dede{l}{a} \diamond a
\right)\cdot{w}\diff{V} \diff{t}
+ 
\left[
\int_{\mathcal{D}}\dede{l}{{u}}\cdot{w}\diff{V}
\right]_{t_0}^{t_1},
\end{equation}
where $\ad_u^*$ is the dual operator to $\ad_u$ defined by
\[
\int_{\mathcal{D}} v\cdot\ad_u^*m\diff{V} = \int_{\mathcal{D}}m\cdot\ad_uv\diff{V},
\]
for all vector fields $v$, and whose explicit formula in components is
\[
\ad^*_u m = \nabla\cdot\left({u}\otimes m\right)
+ (\nabla{u})^Tm
\,.
\]
This formula also happens to match the components of the Lie derivative for one-form densities \cite{HoMaRa1998},
\[
\mathcal{L}_u(m\cdot\diff{x}\otimes\diff{V})
= \left(\nabla\cdot\left({u}\otimes m\right)
+ (\nabla{u})^Tm\right)\cdot\diff{x}\otimes\diff{V},
\]
with line element $\diff{x}$ and volume element $\diff{V}$.
Notation for the diamond operation $(\,\diamond\,)$ has also been introduced in equation (\ref{HP1}). The diamond operation is defined by
\begin{equation}\label{diamond-def}
\int_{\mathcal{D}}  \left(\dede{l}{a} \diamond a \right)\cdot{w}\diff{V} 
:=
\int_{\mathcal{D}}  \dede{l}{a} \cdot \left( - \mathcal{L}_w{a} \right) \diff{V} 
\,.
\end{equation}
Vanishing of the first term in (\ref{HP1}) for variations
that are otherwise arbitrary now produces the Euler-Poincar\'e (EP)
equation,

\begin{equation}\label{EP1}
\pp{}{t}\dede{l}{{u}} + \nabla\cdot\left({u}\otimes\dede{l}{{u}}\right)
+ (\nabla{u})^T\dede{l}{u} - \dede{l}{a} \diamond a
=
0
\,.
\end{equation}

The EP equation in (\ref{EP1}) is completed as an evolutionary system by including the equation of motion (\ref{advect1}) for the advected quantities, $a$.
\paragraph{Noether's Theorem for Euler-Poincar\'e with advected quantities} 
We consider symmetries of the action $S=\int l(u,a)\,dt$ that are
obtained by infinitesimal transformations of the form $\delta g = \eta
\circ g$ for a vector field $\eta$. Consequently, we have the infinitesimal transformations
\begin{equation}\label{inf-trans}
\delta{u} = {\dot \eta}+[u,\eta]\,, \quad
\delta a = -\mathcal{L}_{\eta}a\,.
\end{equation}
If the vector field ${\eta}$ generates symmetries of
the Lagrangian, then Hamilton's principle $\delta S=0$ implies that
\begin{equation}
0 = - \int_{t_0}^{t_1}\!\! \int_{\mathcal{D}} \underbrace{\left(
\pp{}{t}\dede{l}{{u}} + \nabla\cdot\left({u}\otimes\dede{l}{{u}}\right)
+ (\nabla{u})^T\dede{l}{u} - \dede{l}{a} \diamond a
\right)}_{\hbox{= 0}}\cdot\,{\eta}\diff{V} \diff{t}
+ 
\left[
\int_{\mathcal{D}}\dede{l}{{u}}\cdot{\eta}\diff{V}
\right]_{t_0}^{t_1}.
\end{equation}
Here, the term in the time integral vanishes for solutions of the
Euler-Poincar\'e equations, and we are left with the endpoint terms
for arbitrary ${t_0}$ and ${t_1}$. This implies Noether's theorem. 
Namely, a conservation law is associated with each vector field
${\eta}$ that generates a symmetry of the Lagrangian
\cite{Noether1918}. These considerations prove the following.
\begin{theorem}[Noether theorem for EP]
\label{noether ep} 
  Each symmetry vector field $\eta$ of the EP Lagrangian (\ref{action1}) 
  for infinitesimal transformations given by (\ref{inf-trans}) corresponds to
  an integral of the EP motion equation (\ref{EP1}) satisfying
\begin{equation}
\label{e:noether}
\dd{}{t}\int_{\mathcal{D}}\dede{l}{{u}}\cdot{\eta}\diff{V}=0\,,
\end{equation}
for an appropriate inner product. 
\end{theorem}

\subsection{Relabelling symmetries}
Let us now consider how to derive the vector fields $\eta$ for the
symmetry transformations in Noether's theorem in the case of fluids in
the \emph{Eulerian} representation. These symmetry transformations are called
relabelling symmetries. They arise from the assumed right invariance
of the EP Lagrangian $l(u,a)$ under the group $G$ of 
diffeomorphisms (the Lie group of smooth invertible maps with smooth
inverses). 
The Eulerian velocity $u(t) :=\dot{g}g(t)^{-1}\in\mathfrak{X}(\mathbb{R}^3)$ is right-invariant under this
action and therefore it does not change under relabelling
transformations. 
This invariance implies the following evolution
equation for the vector field ${\eta}$:
\begin{equation}
\label{e:eta_vector_field}
\delta u
=
{\dot{\eta}} + [{u},{\eta}] = 0,
\end{equation}
where the bracket $[\,\cdot\,,\,\cdot\,]$ denotes commutation of vector fields.
If a set of advected quantities $\{a\}$ exists, then the vector fields ${\eta}$ for the symmetry transformations must also satisfy the additional conditions that
\begin{equation}
\label{e:eta_advected_condition}
\delta a
=
-\,\mathcal{L}_{{\eta}}a=0\,,
\end{equation}
for each advected quantity $a$. 

When there are no advected quantities present (as in the case of EPDiff \cite{EPDiff}, for example) equation \eqref{e:noether} simply recovers the equation for conservation of momentum, as one sees from the following direct computation:
\begin{eqnarray*}
0 & = & \dd{}{t}\left\langle\dede{l}{{u}},{\eta}\right\rangle \\
 & = & \left\langle \dd{}{t}\dede{l}{{u}}, {\eta} \right\rangle
+ \left\langle \dede{l}{{u}}, \dd{}{t}{\eta} \right\rangle \\
 & = & \left\langle \dd{}{t}\dede{l}{{u}}, {\eta} \right\rangle
+ \left\langle \dede{l}{{u}}, -[{u},{\eta}] \right\rangle \\
 & = & \left\langle \dd{}{t}\dede{l}{{u}}, {\eta} \right\rangle
+ \left\langle \dede{l}{{u}}, \ad_{{u}}{\eta} \right\rangle \\
 & = & \left\langle \dd{}{t}\dede{l}{{u}} + \ad^*_{{u}}
\dede{l}{{u}}, {\eta} \right\rangle 
\\
& = &  
 \left\langle \left(\pp{}{t} + \mathcal{L}_{u(t)} \right)
\dede{l}{u}, {\eta} \right\rangle
.
\end{eqnarray*}
In this computation, the angle brackets $\langle\,\cdot\,,\,\cdot\,\rangle$ denote the $L^2$ pairing $\mathfrak{X}^*\times\mathfrak{X}\to \mathbb{R}$ between the vector fields and their $L^2$ duals, the 1-form densities.

\subsection{Theorems for advected quantities}

We now develop general results for the case where one
or more advected quantities are present. This requires determining 
whether all the conditions in \eqref{inf-trans} can be satisfied
simultaneously. We shall conclude that, if they are satisfied
initially, then they are satisfied for all times $t$, due to the
commutative properties of Lie derivatives. This will enable us to
derive conservation laws in various cases in the rest of the paper.

\begin{theorem}[Commutator]\label{com-thm}
For any pair of smooth time-dependent vector fields $u(t),\eta(t)\in\mathfrak{X}$ and for any $a(t)\in V$ the following commutation relation holds among Lie derivatives,
\begin{equation}
\left[\partial_t + \mathcal{L}_{u(t)}\,,\,\mathcal{L}_{\eta(t)}\right]a(t)
=
\mathcal{L}_{({\dot{\eta}} + [{u},{\eta}])}a(t)
\,.
\label{comrel}
\end{equation}
\end{theorem}

\begin{proof}
For any $a(t)\in V$, one computes by the product rule for Lie derivatives that
\begin{eqnarray*}
\left(\partial_t + \mathcal{L}_{u(t)}\right)\mathcal{L}_{\eta}a(t)
&=&
\mathcal{L}_{({\dot{\eta}} + [{u},{\eta}])}a(t)
+
\mathcal{L}_{\eta}\left(\partial_t + \mathcal{L}_{u(t)}\right)a(t)
\,.
\end{eqnarray*}
Hence, the commutation relation in (\ref{comrel}) holds, and because $a(t)\in V$ is arbitrary,  this implies the Lie derivative commutation relation 
\begin{equation}
\left[\partial_t + \mathcal{L}_{u(t)}\,,\,\mathcal{L}_{\eta}\right]
=
\mathcal{L}_{({\dot{\eta}} + [{u},{\eta}])}
\,.
\label{comrel-proof}
\end{equation}
\end{proof}
 
Under the assumption that the variational vector field ${\eta}$ satisfies the time-evolution equation  \eqref{e:eta_vector_field} required for a relabelling symmetry transformation, one finds the following commutator theorem. 
 
\begin{corollary}[Symmetry] \label{sym-cor}
If a vector field $\eta$ satisfies equation (\ref{e:eta_vector_field}) for an infinitesimal relabelling symmetry, then the Lie derivative $\mathcal{L}_\eta$ commutes with the evolution operator, $\left(\partial_t + \mathcal{L}_{u(t)}\right)$,
\begin{equation}
\left[\partial_t + \mathcal{L}_{u(t)}\,,\,\mathcal{L}_{\eta(t)}\right]a(t)
=
0
\quad\hbox{for}\quad
{\dot{\eta}} + [{u},{\eta}] = 0
\,.
\label{symrel}
\end{equation}
\end{corollary}

\begin{proof}
This symmetry corollary follows by inserting equation (\ref{e:eta_vector_field}) into the commutation relation in equation  \eqref{comrel}.
\end{proof}

\begin{theorem}[Ertel theorem]\label{Ertel-thm}
If the quantity $a$ is advected as in equation (\ref{advect1}) and the vector field $\eta$ satisfies equation (\ref{e:eta_vector_field}) for an infinitesimal relabelling symmetry, then $\mathcal{L}_\eta{a}$ is also advected.
\end{theorem}

\begin{proof}
By equations (\ref{advect1}) and  \eqref{e:eta_vector_field} one finds the advection relation for $\mathcal{L}_\eta{a}$,
\begin{equation}
\left(\partial_t + \mathcal{L}_{u(t)}\right)\mathcal{L}_{\eta}a(t)
=
\mathcal{L}_{\eta}\left(\partial_t + \mathcal{L}_{u(t)}\right)a(t)
=
0
\,,
\label{persist-eqn}
\end{equation}
as a result of  the condition (\ref{advect1}) satisfied by advected quantities.
\end{proof}

\noindent
Consequently, if $\mathcal{L}_\eta a=0$ in equation \eqref{e:eta_advected_condition} holds initially, then it continues to hold under the EP flow. That is, we have the following.

\begin{corollary}[Persistence]\label{Persistence-cor}
If the vector field $\eta$ is a relabelling symmetry, then the symmetry condition for advected quantities $\mathcal{L}_{\eta}a(t)=0$ persists, provided it holds initially.
\end{corollary}

\begin{proof}
If the left side of equation (\ref{persist-eqn}) vanishes initially at $t=0$, then it continues to vanish for all time $t>0$.
\end{proof}
\begin{definition}[Locally conserved quantities]
A locally conserved quantity $c(t)$ follows from the equations of motion and satisfies a local conservation law,
\begin{equation}
\left(\partial_t + \mathcal{L}_{u(t)}\right)c(t) = 0
\,,
\label{localCL}
\end{equation}
which has the same form as an advection law.
\end{definition}

\begin{remark}[Local conservation laws]\rm
One distinguishes between advected quantities and locally conserved quantities. 
Namely, equations for advected quantities are obtained from the action $G\times V\to V$ and are independent of the fluid velocity. In contrast, local conservation laws involve the fluid velocity because they arise from the equations of motion. 
\end{remark}

\begin{corollary}[Iterated conserved quantities]
If the quantity $c(t)$ satisfies a local conservation law (\ref{localCL}) as a result of the EP equations of motion for all time, then $\mathcal{L}_{\eta}c(t)$ is also locally conserved for any relabelling symmetry.
\end{corollary}

\begin{proof}
This follows from replacing $a(t)$ by $c(t)$ in equation (\ref{persist-eqn}).   
\end{proof}

\begin{remark}
Iterating this process further is possible, but once a conserved quantity can be expressed in terms of advected quantities, iteration does not lead to new information.
\end{remark}
In the following section we adopt the strategy of constructing the symmetry vector fields $\eta$ in Noether's theorem using the advected quantities that they preserve. This is accomplished in the examples below for several representative fluid flows in three dimensions. Occasionally, when enough freedom remains in the infinitesimal Lie symmetry $\eta$, the local conservation law that emerges from Noether's theorem may be re-substituted into the weak form of Noether's theorem to compute an additional integral conservation law. 
\section{Examples} \label{sec-examples}
\subsection{Advected density: Conservation of vorticity and helicity}
For the specific case that the mass density $a=\rho\diff{V}$ is
advected and other advected quantities are absent, e.g., in barotropic
fluid dynamics, the symmetry condition
(\ref{e:eta_advected_condition}) is
\[
\mathcal{L}_{\eta}(\rho\diff{V})
= d\left(\eta\contract\rho\diff{V}\right)
=0
\,.
\]
Therefore, by Poincar\'e's Lemma, one may write 
locally that
\begin{equation}
\eta\contract\rho\diff{V}
=
d({\boldsymbol\Psi}\cdot\diff{\MM{x}})
=
\curl{\boldsymbol\Psi}\cdot\diff{\MM{S}}
\,,
\label{Psi-2form}
\end{equation}
for some vector function ${\boldsymbol\Psi}$. 

For non-trivial topology
(on a spherical annulus, for example), we may choose a simply connected
patch bounded by a simple closed curve $C(t)$ that is transported by
the fluid velocity $u$. We then restrict $\eta$ at each time to the
Lie algebra of vector fields that leave $C(t)$ invariant.
(These vector fields $\eta$ are tangent to the curve $C(t)$.) 
This choice allows us to define ${\boldsymbol\Psi}$ on the patch enclosed by $C(t)$
for each relabelling symmetry $\eta$.

Equation (\ref{Psi-2form}) the vector field $\eta$ for each relabelling symmetry may be expressed in terms of a vector function ${\boldsymbol\Psi}$, as
\[
\eta = \rho^{-1}{\curl}{\boldsymbol\Psi}\cdot\nabla
.\]
All such vector fields satisfy the advection condition (\ref{e:eta_advected_condition}) for the density 
$\rho\diff{V}$, since
\[
\mathcal{L}_{\eta}(\rho\diff{V}) 
= {\rm div} (\rho\,\rho^{-1}\curl{\boldsymbol\Psi})\,\diff{V} 
= {\rm div} (\curl{\boldsymbol\Psi})\,\diff{V} = 0
\,.
\]

We substitute this solution for $\eta$ into Noether's theorem and use Corollary \ref{Persistence-cor} (persistence  of the symmetry relation) to find,
\begin{eqnarray*}
  0  & = & \dd{}{t}\left\langle\dede{l}{u},{\eta}\right\rangle 
  \\
& = & \dd{}{t}\int_{\mathcal{D}}
\dede{l}{u}\cdot {\eta}\diff{V}
  \\
  & = & \dd{}{t}\int_{\mathcal{D}}
  \frac{1}{\rho}\dede{l}{\MM{u}}\cdot \diff{\MM{x}}\wedge{\eta}\contract(\rho\diff{V})
  \\
 \hbox{By (\ref{Psi-2form})}
  & = & \dd{}{t}\int_{\mathcal{D}}
  \frac{1}{\rho}\dede{l}{\MM{u}}\cdot \diff{\MM{x}}\wedge
  \diff({\boldsymbol\Psi}\cdot\diff{\MM{x}})
  \\
  & = & \int_{\mathcal{D}}
  \pp{}{t}\left(\frac{1}{\rho}\dede{l}{\MM{u}}\cdot \diff{\MM{x}}\right)\wedge\diff
  ({\boldsymbol\Psi}\cdot\diff{\MM{x}})
  +
  \int_{\mathcal{D}}
  \frac{1}{\rho}\dede{l}{\MM{u}}\cdot \diff{\MM{x}}\wedge\pp{}{t}\diff({\boldsymbol\Psi}\cdot\diff{\MM{x}})
   \\
  \hbox{By (\ref{persist-eqn}) }
  & = & \int_{\mathcal{D}}
  \pp{}{t}\left(\frac{1}{\rho}\dede{l}{\MM{u}}\cdot \diff{\MM{x}}\right)\wedge\diff
  ({\boldsymbol\Psi}\cdot\diff{\MM{x}})
  +
  \int_{\mathcal{D}}
  \frac{1}{\rho}\dede{l}{\MM{u}}\cdot \diff{\MM{x}}\wedge(-\diff\mathcal{L}_u({\boldsymbol\Psi}\cdot
  \diff{\MM{x}}))
   \\
  & = & -\int_{\mathcal{D}} 
 \left( \pp{}{t}\diff\left(\frac{1}{\rho}\dede{l}{\MM{u}}\cdot \diff{\MM{x}}\right)
  +
  \mathcal{L}_u
  \diff
\left(\frac{1}{\rho}\dede{l}{\MM{u}}\cdot \diff{\MM{x}}\right)\right)
\wedge
  ({\boldsymbol\Psi}\cdot\diff{\MM{x}})
  \\
& = &  
-\int_{\mathcal{D}} 
 \left(\pp{}{t} + \mathcal{L}_{u(t)} \right)
  \diff
\left(\frac{1}{\rho}\dede{l}{\MM{u}}\cdot \diff{\MM{x}}\right)
\wedge
  ({\boldsymbol\Psi}\cdot\diff{\MM{x}})
\,.
\end{eqnarray*}
Since ${\boldsymbol\Psi}$ is arbitrary, vanishing of the final line gives the weak form of the \bfi{local conservation of the vorticity 2-form}, 
\begin{equation}
\left(\pp{}{t} +
  \mathcal{L}_u\right)
\left(\curl\frac{1}{\rho}\dede{l}{\MM{u}}\cdot \diff{\MM{S}}\right)
=0
\,,
\label{vort2form}
\end{equation}
where vorticity is defined as the curl of the specific momentum (momentum per unit mass). The specific momentum is equal to the velocity for Euler fluids; so its curl in that case is the usual Euler fluid vorticity.

\begin{remark}[Two velocity vectors]\rm
Two velocity vectors appear in the computations above: These are the
fluid velocity vector $\MM{u}$ in the Lie derivative $\mathcal{L}_u$
and the specific momentum covector $\rho^{-1}\delta l/\delta\MM{u}$ in
the 1-form $\rho^{-1}\delta l/\delta\MM{u} \cdot \diff\MM{x}$. These
two velocities are the basic ingredients for performing  modelling and
analysis in any ideal fluid problem. They appear together and have
separate meanings in the Euler-Poincar\'e equation and throughout the
present paper, as illustrated in the examples below.
\end{remark}
\begin{example}[Incompressible Euler equations]
The incompressible Euler equations have reduced Lagrangian 
\[
l(u,\rho) = \int_{\mathcal{D}} \frac{\rho|u|^2}{2} + p(1-\rho)\diff{V},
\]
where $p$ is a Lagrange multiplier enforcing incompressibility $\rho=1$.
The variational derivatives are 
\[
\dede{l}{u} = \rho u, \quad \dede{l}{\rho} = \frac{|u|^2}{2} - p.
\]
In this case, the conserved vorticity is 
\[
\curl\left(\frac{1}{\rho}\dede{l}{u}\right)  = \curl u.
\]
\end{example}
\begin{example}[Incompressible Euler-alpha equations]
The incompressible Euler-alpha equations with $\rho=1$ have reduced Lagrangian 
\[
l(u,\rho) = \int_{\mathcal{D}} \frac{\rho}{2}(|u|^2+\alpha^2|\nabla u|^2) + p(1-\rho)\diff{V}.
\]
The variational derivatives are 
\[
\dede{l}{u} = \rho u - \alpha^2\nabla\cdot \rho \nabla u, \quad
\dede{l}{\rho} = \frac12(|u|^2+\alpha^2 |\nabla u|^2) - p.
\]
In this case, the conserved vorticity is 
\[
\curl\left(\frac{1}{\rho}\dede{l}{u}\right)  = \curl u - \frac{\alpha^2}{\rho}\nabla\cdot \rho \nabla u,
\]
which becomes $\curl(u - \alpha^2\nabla^2u)$ since $\rho=1$.
\end{example}
In the preserved 2-form $\diff({\boldsymbol\Psi}\cdot\diff{\MM{x}})$ introduced in equation (\ref{Psi-2form}), the 
vector function ${\boldsymbol\Psi}$ is determined (locally) for each choice of symmetry vector field $\eta$. 
Likewise, we have seen that each choice of the vector function ${\boldsymbol\Psi}$ corresponds to a certain relabelling symmetry $\eta$. 
Consequently,  Corollary \ref{Persistence-cor} for the persistence of symmetry and the definition of a local conservation law in equation (\ref{localCL}) would allow us to replace $\diff({\boldsymbol\Psi}\cdot\diff{\MM{x}})$ with another conserved 2-form. In particular, we may choose the conserved vorticity 2-form and set,
\[
\diff({\boldsymbol\Psi}\cdot\diff{\MM{x}})= 
\diff\left(\frac{1}{\rho}\dede{l}{\MM{u}}\cdot \diff{\MM{x}}\right)
.\]
After this identification, we may draw the conclusion from persistence in Corollary \ref{Persistence-cor} that
\begin{eqnarray*}
  0 & = & \dd{}{t}\left\langle\dede{l}{u},{\eta}\right\rangle 
  \\
  & = & \dd{}{t}\int_{\mathcal{D}}
  \frac{1}{\rho}\dede{l}{\MM{u}}\cdot \diff{\MM{x}}\wedge{\eta}\contract\rho\diff{V}
  \\
  & = & \dd{}{t} \int_{\mathcal{D}}
  \left(\frac{1}{\rho}\dede{l}{\MM{u}}\cdot \diff{\MM{x}}\right)\wedge\diff
  ({\boldsymbol\Psi}\cdot\diff{\MM{x}})
  \\
  & = & \dd{}{t}\int_{\mathcal{D}}
  \frac{1}{\rho}\dede{l}{\MM{u}}\cdot \diff{\MM{x}}\wedge
\diff\left(\frac{1}{\rho}\dede{l}{\MM{u}}\cdot \diff{\MM{x}}\right)
  \\
  & = & \dd{}{t}\int_{\mathcal{D}}
  \frac{1}{\rho}\dede{l}{\MM{u}}\cdot 
\curl\left(\frac{1}{\rho}\dede{l}{\MM{u}}\right)
\diff{V}
\,.
\end{eqnarray*}
Thus, the weak form of the local conservation law for vorticity yields a  conservation law for the \bfi{helicity} integral,
\begin{equation}
H := \int_{\mathcal{D}}
  \frac{1}{\rho}\dede{l}{\MM{u}}\cdot 
\curl\left(\frac{1}{\rho}\dede{l}{\MM{u}}\right)
\diff{V}
\,.
\label{helicity-int}
\end{equation}
The vector field for the symmetry associated with helicity conservation is,
\begin{equation}
\eta_{H}
=   
{\rho}^{-1}\curl\left(\frac{1}{\rho}\dede{l}{\MM{u}}\right)
\cdot\nabla
.
\label{helicityVF}
\end{equation}
The characteristic paths of the vector field $\eta_{H}$ may be regarded as vortex lines, and these  satisfy the symmetry condition (\ref{e:eta_vector_field}), as a result of the EP equation (\ref{EP1}). That is, the characteristic paths of $\eta_{H}$ are frozen into the flow of the fluid velocity. This means that shifts along these paths are relabelling symmetries and the corresponding Noether conservation law is the helicity ${H}$ in equation (\ref{helicity-int}).
In particular, the symmetry associated with conservation of helicity is a relabelling of the frozen-in vortex lines. 

As mentioned earlier, the momentum per unit mass  is equal to the velocity $\MM{u}$ for Euler fluids, so that conservation of the helicity for Euler fluids may be expressed as
\[
\dd{}{t}\int_{\mathcal{D}}
  \MM{u}\cdot \curl\,\MM{u}
\diff{V}
=
0
\,.
\]
The spatial integral $H$ defining the fluid helicity in (\ref{helicity-int}) measures the knottedness, or number of linkages, of the vortex lines, that is, lines of $\MM{\omega}=\curl({\rho}^{-1}\delta{l}/\delta\MM{u})$. This fluid helicity indicates the topological complexity of the winding of the vortex lines in $\MM{\omega}$ amongst themselves in the spatial domain \cite{ArKh1998}. 
Physically, helicity conservation arises because the vortex lines are frozen into the flow of the diffeomorphisms and, thus, they cannot unknot.

\begin{remark}[Ertel's theorem in hydrodynamic notation]\rm
We identify the evolution operator in Theorem \ref{Ertel-thm} as the familiar Lagrangian time derivative ${\rm D}/{\rm D}t$,
\[
\partial_t + \mathcal{L}_{u(t)}  = \frac{{\rm D}}{{\rm D}t}
\,,
\]
and we express the vector field $\eta_{H}$ in equation (\ref{helicityVF}) in terms of a generalised vorticity vector $\MM{\omega}$, defined as
\[
\eta_{H}
=   
{\rho}^{-1}\curl\left(\frac{1}{\rho}\dede{l}{\MM{u}}\right)
\cdot\nabla
=:
{\rho}^{-1}\MM{\omega}\cdot\nabla
\quad\hbox{with}\quad
\MM{\omega} 
:= 
\curl\left(\frac{1}{\rho}\dede{l}{\MM{u}}\right)
.
\]
Introducing this familiar hydrodynamic notation allows one to write the symmetry relation (\ref{symrel}) in the case for the action of the Lie derivative on a function $a(t)$ as 
\begin{equation}
\frac{{\rm D}}{{\rm D}t}({\rho}^{-1}\MM{\omega}\cdot\nabla) a(t)
=
({\rho}^{-1}\MM{\omega}\cdot\nabla) \frac{{\rm D}}{{\rm D}t}a(t)
\,,
\label{Ertel-D/Dt}
\end{equation}
which is the usual form of the classical Ertel theorem \cite{{Er1942}}. 
For a scalar advected function, $a\in \Lambda^0$,  Corollary \ref{Persistence-cor} (persistence)  yields yet another scalar conservation law, for $q=({\rho}^{-1}\MM{\omega}\cdot\nabla)a$.
\end{remark}

\subsection{Advected density and tracer: Conservation of potential vorticity}
\begin{proposition}
For the case of two advected quantities $a_1=\rho\diff{V}\in \Lambda^3$, $a_2={s} \in \Lambda^0$, the simultaneous solution of $\mathcal{L}_\eta\rho\diff{V}=0$ and $\mathcal{L}_\eta {s}=0$ is
\begin{equation}
\eta\contract\rho\diff{V} = \diff\left(\phi\diff {s}\right),
\label{simulsoln}
\end{equation}
for general $\phi\in \Lambda^0$.
\end{proposition}

\noindent
The proof of this proposition is simple, because $a_1$ is a top form and $a_2$ is a bottom form.

\begin{proof}Symmetry requires that these two advected quantities satisfy
\[
\eta\contract\rho\diff{V} = \diff ({\boldsymbol\Psi}\cdot\diff{\MM{x}})
\quad\hbox{and}\quad
\eta\contract d{s} =  \eta \cdot \nabla {s} =  0
\,.
\]
Thus,
\[
d{s}\wedge \eta\contract\rho\diff{V}  =  (\nabla {s}\cdot \eta)\, \rho\diff{V}= 0
\,,
\]
and, hence,
\[
0 = (\nabla {s}\cdot \eta)\, \rho\diff{V} =
d{s}\wedge (\eta\contract\rho\diff{V})  =  
d{s}\wedge \diff ({\boldsymbol\Psi}\cdot\diff{\MM{x}})  =
d{s}\wedge \diff\left(\phi\diff {s}\right)
.\]
\end{proof}
The two advected quantities $d({\boldsymbol\Psi}\cdot\diff{\MM{x}})$ and $\diff\left(\phi\diff {s}\right)$ both equal $(\eta\contract\rho\diff{V})$, so they satisfy the same evolution equation. In particular, the following advection equation holds
\[
\pp{}{t}\diff\left(\phi\diff {s}\right)
= -\diff\mathcal{L}_u\left(\phi\diff {s}\right).
\]
Now one may substitute $\eta\contract\rho\diff{V}=\diff(\phi\diff {s})$ into the Noether theorem calculation as above and recompute, finding this time that:
\begin{eqnarray*}
  0 & = & \dd{}{t}\left\langle\dede{l}{u},{\eta}\right\rangle 
  \\
  & = & 
  \dd{}{t}\int_{\mathcal{D}}\left(\frac{1}{\rho}\dede{l}{\MM{u}}\cdot \diff{\MM{x}}\right)
\wedge\,(\eta\contract\rho\diff{V})
  \\
  \hbox{By  (\ref{simulsoln})}
  & = & 
  \dd{}{t}\int_{\mathcal{D}}\left(\frac{1}{\rho}\dede{l}{\MM{u}}\cdot \diff{\MM{x}}\right)
\wedge\, \diff(\phi\diff {s})
\\
  \hbox{By (\ref{persist-eqn})}
  & = & 
  \int_{\mathcal{D}}
  \pp{}{t}\left(\frac{1}{\rho}\dede{l}{\MM{u}}\cdot \diff{\MM{x}}\right)
  \wedge
  \diff(\phi\diff {s})
  +
  \int_{\mathcal{D}}
  \frac{1}{\rho}\dede{l}{\MM{u}}\cdot \diff{\MM{x}}
  \wedge
  (-\diff\mathcal{L}_u(\phi\diff {s}))
   \\
  & = & -\left\langle
  \pp{}{t}\diff\left(\frac{1}{\rho}\dede{l}{\MM{u}}\cdot \diff{\MM{x}}\right)
  \wedge\diff {s}
  +
  \mathcal{L}_u
  \diff\left(\frac{1}{\rho}\dede{l}{\MM{u}}\cdot \diff{\MM{x}} \right)\wedge \diff {s},\phi
  \right\rangle \\
  \hbox{By ${s}=a_2$}
  & = & -\left\langle
  \pp{}{t}\left(\diff\left(\frac{1}{\rho}\dede{l}{\MM{u}}\cdot \diff{\MM{x}}\right)
  \wedge\diff {s}\right)
  +
  \mathcal{L}_u\left(
  \diff\left(\frac{1}{\rho}\dede{l}{\MM{u}}\cdot \diff{\MM{x}} \right)\wedge \diff {s}
  \right),\phi
  \right\rangle
\\
& = &  
- \left\langle \left(\pp{}{t} + \mathcal{L}_{u(t)} \right)
\left(
  \diff\left(\frac{1}{\rho}\dede{l}{\MM{u}}\cdot \diff{\MM{x}} \right)\wedge \diff {s}
  \right), {\phi} \right\rangle
.
\end{eqnarray*}
As before, all boundary terms have been dropped in spatial integrations by parts.
Since $\phi$ is arbitrary, the final line of the calculation above gives the weak form of the conservation law for Ertel \bfi{potential vorticity} (PV) density, defined as \cite{Er1942}, 
\begin{equation}
q\,\rho\diff{V}
:=
\diff\left(\frac{1}{\rho}\dede{l}{\MM{u}}\cdot \diff{\MM{x}}\right) \wedge \diff {s}
=
\curl\left(\frac{1}{\rho}\dede{l}{\MM{u}}\right)\cdot
\nabla{s} \,dV
\,.
\label{pv-def}
\end{equation}
The corresponding local conservation law is
\begin{equation}
\left(\pp{}{t} +
  \mathcal{L}_u\right)
\left(q\,\rho\diff{V}\right)
=0
\,.
\label{pv-advect}
\end{equation}
The arbitrary function $\phi$ in the weak form of the local conservation law for potential vorticity also yields the integral conservation law, 
\begin{equation}
\dd{}{t}\int_{\mathcal{D}} \Phi(q)\,\rho\diff{V} = 0
\,,
\label{pv-int}
\end{equation}
in which $\Phi$ is an arbitrary function 
and we used $\partial_t(\Phi(q)\rho)=-\,\nabla\cdot(\Phi(q)\rho)$.

The vector field for the symmetry associated with PV conservation in (\ref{pv-advect}) may be computed from equation (\ref{simulsoln}) as
\begin{equation}
\eta_{PV}
=   
{\rho}^{-1}\left(\nabla\phi\times\nabla{s}\right)
\cdot\nabla
,
\label{pv-eta}
\end{equation}
and it represents shifts along level sets of ${s}$. In particular, the
characteristic paths of the vector field $\eta_{PV}$ lie along the
level sets of ${s}$, which in turn are frozen into the flow of the fluid
velocity. This means that shifts along the characteristic paths of
$\eta_{PV}$ are relabelling symmetries and the corresponding Noether
conservation law is the advection of the potential vorticity $q$ in
equation (\ref{pv-def}).

\begin{example}[Rotating Euler-Boussinesq equations]
The reduced Lagrangian for the rotating Euler-Boussinesq equations is
\[
l = \int_{\mathcal{D}} \rho\frac{|u|^2}{2}+ \rho u \cdot R - z b + p(1-\rho)\diff{V},
\]
where $b$ is the buoyancy satisfying $\pp{b}{t} + \mathcal{L}_ub=0$,
and where $R$ satisfies $\curl R = 2\MM{\Omega}$. Consequently,   the 
conserved potential vorticity may be computed, as follows:
  \[
  \curl\left(\frac{1}{\rho}\dede{l}{\MM{u}}\right)
  =
  \curl(\MM{u}+\MM{R})
  \quad\hbox{with}\quad
  \curl\MM{R} =  2\MM{\Omega}
  \quad\Longrightarrow\quad
  q = (\curl\MM{u} + 2\MM{\Omega} ) \cdot\nabla{b}.
  \]
\end{example}

\subsection{Advected density and flux (2-form): Conservation of cross helicity}

\begin{proposition}
  For the case that $a_1=\rho\diff{V}\in \Lambda^3$ and
  $a_2=\MM{B}\cdot\diff \MM{S} = \diff (\mathbf{A}\cdot\diff{\MM{x}})
  \in \Lambda^2$, the only
simultaneous
  solution of $\mathcal{L}_\eta\rho\diff{V}=0$ and $\mathcal{L}_\eta
  \MM{B}\cdot\diff \MM{S}=0$ is
\begin{equation}
\eta\contract\rho\diff{V} = \MM{B}\cdot\diff \MM{S}
\,.
\label{simulsoln1}
\end{equation}
\end{proposition}

\begin{proof}
\[
\hbox{Recall}\quad
\eta\contract\rho\diff{V} = \diff ({\boldsymbol\Psi}\cdot\diff{\MM{x}})
\quad\hbox{and identify}\quad
\diff ({\boldsymbol\Psi}\cdot\diff{\MM{x}})
=  \MM{B}\cdot\diff \MM{S}
\,.
\]
\end{proof}

\noindent
In this case, Noether's theorem implies the conserved quantity
\begin{eqnarray*}
  0 & = & \dd{}{t}\left\langle\dede{l}{\MM{u}},\MM{\eta}\right\rangle 
  \\
  & = & 
  \dd{}{t}\int_{\mathcal{D}}\left(\frac{1}{\rho}\dede{l}{\MM{u}}\cdot \diff{\MM{x}}\right)
\wedge\,(\eta\contract\rho\diff{V})
\\
  \hbox{By  (\ref{simulsoln1})}
  & = & 
    \dd{}{t}\int_{\mathcal{D}} \frac{1}{\rho}\dede{l}{\MM{u}}\cdot \diff{\MM{x}} 
\wedge\,\MM{B}\cdot\diff \MM{S}
 \\
  & = & 
    \dd{}{t}\int_{\mathcal{D}} \left(
    \MM{B}\cdot \frac{1}{\rho}\dede{l}{\MM{u}}
    \right)
   \diff{V}
    \,.
\end{eqnarray*}
This is the \bfi{cross helicity}, which is known to be conserved, in particular, for ideal magnetohydrodynamics (MHD) \cite{ArKh1998}. 
The vector field for the symmetry associated with conservation of cross helicity is,
\[
\eta_{CH}
=   
{\rho}^{-1}\MM{B}
\cdot\nabla
,
\] 
which represents a field of shifts along magnetic field lines. The characteristic paths of the vector field $\eta_{CH}$ are magnetic field lines that satisfy the symmetry condition (\ref{e:eta_vector_field}), as a result of the advection equation (\ref{e:eta_advected_condition}) for magnetic flux $\MM{B}\cdot\diff \MM{S}$. That is, the characteristic paths of $\eta_{CH}$ are frozen into the flow of the fluid velocity. This means that shifts along these characteristic paths are relabelling symmetries and the corresponding Noether conservation law is the cross helicity. 

\section{Other conservation laws for ideal fluids}\label{nonNoetherCLs}

In this section, we first point out that not all fluid conservation laws follow from Noether's theorem, as formulated above, by considering the counterexample of magnetic helicity for MHD. We then make a connection between the Noether's theorem discussed in this paper, and the Kelvin-Noether circulation theorem discussed in \cite{HoMaRa1998}.

\subsection{Magnetic helicity}
The distinction between advected quantities and locally conserved
quantities comes back into play, when one considers compound advected
quantities that are conserved independently of the motion
equation. For example advection of the scalar ${s}$ and the exact 2-form
$\MM{B}\cdot\diff \MM{S}=\diff(\mathbf{A}\cdot\diff \MM{x})$ lead
immediately to advection of the compound quantities,
\[
d{s}\wedge \MM{B}\cdot\diff \MM{S}
= {\rm div}({s} \MM{B}) \diff{V}
\quad\hbox{and}\quad
\mathbf{A}\cdot\diff \MM{x}\wedge \MM{B}\cdot\diff \MM{S}
=
\mathbf{A}\cdot \MM{B} \diff{V}
\,.
\]
The former is trivial, because its integral over space vanishes
identically. However, the latter is the famous \bfi{magnetic
  helicity}, whose spatial integral measures the knottedness, or
number of linkages, of the magnetic field lines. That is, the magnetic
helicity indicates the topological complexity of the winding of the
magnetic field lines amongst themselves in the spatial domain. The
preservation of this magnetic winding number is a fascinating property
of ideal MHD flows \cite{ArKh1998}, but it does not arise from a
Noether symmetry. It arises here as a compound Lagrangian quantity
whose Eulerian interpretation is deep and interesting. It is beyond
our present scope for further study, except to provide a counter
example to the conjecture that a converse of Noether's theorem might
exist for Euler-Poincar\'e ideal fluid theories. \\

\subsection{Modified vorticity, potential vorticity, helicity and the Kelvin-Noether theorem}

More general conservation laws can be obtained by expanding the set of variables, so that the time variation of the quantity $a$ is enforced by a Lagrange multiplier $b$ (known as a Clebsch variable) instead of constraining the variation $\delta{a}$. As described in
\cite{CoHoHy2007,CoHo2009a,GaBaRa2011}, the same Euler-Poincar\'e equations 
are obtained this way. In this case, Hamilton's principle becomes
\begin{equation}
S_{Clebsch}
=
\int_{t_0}^{t_1}\!\! 
l(u,a) + \left\langle b\,, 
\left(\pp{}{t} +  \mathcal{L}_u\right)a\right\rangle
dt
\,,
\label{Clebsch-Lag}
\end{equation}
with Lagrange multiplier $b$ to be determined.
Then, Hamilton's principle yields after a short calculation,
\begin{eqnarray*}
0 = \delta S_{Clebsch}
&=&
\int_{t_0}^{t_1}\!\! 
\left\langle \dede{l}{u} - b\diamond{a}\,,\,\delta{u}\right\rangle
+ \left\langle \delta b\,, 
\left(\pp{}{t} +  \mathcal{L}_u\right)a\right\rangle
\\
&&\quad 
+ \left\langle 
\dede{l}{a} 
-
\left(\pp{}{t} +  \mathcal{L}_u\right)b\,, 
\delta{a}\right\rangle\,dt
+
\Big[\Big\langle  b\,, \delta{a}\Big\rangle\Big]_{t_0}^{t_1}
\,.
\label{var1-Clebsch-Lag}
\end{eqnarray*}
We now consider symmetries of the form
\[
\delta u = \dot{\eta} + [u,\eta], \quad \delta a = 0, \quad \delta b =0,
\]
for a general relabelling vector field $\eta$ that satisfies $\dot{\eta} + [u,\eta]=0$, but is not constrained to be a symmetry of the quantity $a$. Noether's theorem then leads to
\[
\dd{}{t}\left\langle \frac{1}{\rho}\dede{l}{u},\eta\contract
\rho\diff{V}\right\rangle = 0\,.
\]
Next, we define the $(\,\tilde{\diamond}\,)$ operation in terms of the diamond operation by 
\begin{equation}\label{tilde-diamond-def}
b \diamond a =: \left( \frac{b}{\rho}\,\tilde{\diamond}\,a\right)
\otimes\rho\diff{V}\,.
\end{equation}
The $(\,\tilde{\diamond}\,)$ operation allows one to express a 1-form density as the product of a 1-form and the advected mass density. 

After a calculation similar to that leading  to the result  (\ref{vort2form}), one may write the vanishing of the $\eta$-coefficient in the previous variational equation for $\delta S_{Clebsch}=0$ as
\begin{equation}
\dd{}{t}\oint_{\gamma_t}
\frac{1}{\rho}\, \dede{l}{u}
 - \dd{}{t}\oint_{\gamma_t}\frac{b}{\rho}\,\tilde{\diamond}\,{a}
+
\oint_{\gamma_t}
\left(\frac{1}{\rho}\,\dede{l}{a} \right)\,\tilde{\diamond}\,{a}
-
\oint_{\gamma_t}
\left(\frac{1}{\rho}\left(\partial_{t} 
+  
\mathcal{L}_u\right)b\right)\,\tilde{\diamond}\,{a}
=
0
\,,
\label{var2-Clebsch-Lag}
\end{equation}
which is found after substituting $\left(\pp{}{t} +  \mathcal{L}_u\right)a=0$, as imposed by the $\delta{b}$-variation. In the loop integrals, the closed circuit $\gamma_t$ moves with the flow of the fluid velocity vector field $u$. 

Now we have two choices. Namely, we may either eliminate Lagrange multiplier $b$ by using the variational equation for $b$,  
\begin{equation}
\left(\pp{}{t} +  \mathcal{L}_u\right)b
=
\dede{l}{a} 
\,,
\label{b-eqn}
\end{equation}
or we may keep $b$ as an additional dynamical variable satisfying equation (\ref{b-eqn}). The first choice yields the Kelvin-Noether theorem of \cite{HoMaRa1998}, and the second choice yields an advection equation for a quasi-vorticity vector field, plus the additional equation (\ref{b-eqn}) for $b$. Specifically, in the first choice, the second and fourth terms in equation (\ref{var2-Clebsch-Lag}) cancel, leaving
\begin{equation}
\dd{}{t}\oint_{\gamma_t}
\frac{1}{\rho}\, \dede{l}{u}
+
\oint_{\gamma_t}
\left(\frac{1}{\rho}\,\dede{l}{a} \right)\,\tilde{\diamond}\,{a}
=
0
\,,
\label{circ1}
\end{equation}
which is the Kelvin-Noether theorem \cite{HoMaRa1998} for circulation. 

In the second choice, the third and fourth terms in equation (\ref{var2-Clebsch-Lag}) cancel instead, thereby leaving the following circulation conservation law,
\begin{equation}
\dd{}{t}\oint_{\gamma_t}
\left( \frac{1}{\rho}\, \dede{l}{u}
 - 
 \frac{b}{\rho}\,\tilde{\diamond}\,{a}\right)
 =
 0
 \,,
\label{circ2}
\end{equation}
or, equivalently in vector notation,
\begin{equation}
\dd{}{t}\oint_{\gamma_t}
\MM{\tilde{u}} \cdot \diff{\MM{x}}
=
0
\,,
\quad\hbox{with}\quad
\MM{\tilde{u}} \cdot \diff{\MM{x}}
:=
\left( \frac{1}{\rho}\, \dede{l}{u}
 - 
 \frac{b}{\rho}\,\tilde{\diamond}\,{a}\right)
 .
\label{circ3}
\end{equation}
The price for this circulation conservation law is that the Lagrange multiplier $b$ remains and satisfies equation (\ref{b-eqn}), instead of being eliminated. However, keeping $b$ as a dynamical variable also has the added value
that doing so yields a quasi-vorticity $\MM{\tilde{\omega}}$ involving $b$ that satisfies  the advection law for a 2-form under the flow of the velocity vector field, $u$,
\begin{equation}
\left(\pp{}{t} +
  \mathcal{L}_u\right)
\left(\MM{\tilde{\omega}} \cdot \diff{\MM{S}}\right)
=0
\,,
\quad\hbox{with}\quad
\MM{\tilde{\omega}} \cdot \diff{\MM{S}} 
=
(\curl\,\MM{\tilde{u}})\cdot \diff{\MM{S}}
:=
\diff\left( \frac{1}{\rho}\, \dede{l}{u}
 - 
 \frac{b}{\rho}\,\tilde{\diamond}\,{a}\right)
 .
\label{quasi-vort2form}
\end{equation}
Moreover, the vector field $\rho^{-1}{\tilde{\omega}}$, which is derived via the relation
\[
\rho^{-1}{\tilde{\omega}}\contract \rho\diff{V}=\MM{\tilde{\omega}} \cdot \diff{\MM{S}}
\,,
\]
also satisfies the invariance equation (\ref{e:eta_vector_field}), namely,
\begin{equation}
\pp{}{t} \big(\rho^{-1}{\tilde{\omega}}\big) 
+ \Big[{u}\,,\,\rho^{-1}{\tilde{\omega}}\Big] 
= 0
\,,
\quad\hbox{which means}\quad
  \left[\pp{}{t} + \mathcal{L}_u
  \,,\,
  \mathcal{L}_{ (\rho^{-1}\tilde{\omega})} \right]
  =
  0
  \,,
\label{quasivort-invar_vector_field}
\end{equation}
as demonstrated in equation (\ref{comrel-proof}) in the proof of Theorem \ref{com-thm}.

Consequently, keeping the Lagrange multiplier $b$ as a dynamical variable produces an Ertel theorem of the form (\ref{Ertel-D/Dt}) and yields conservation laws for the corresponding potential vorticity and helicity.

Moreover, these equations apply for essentially \emph{any} fluid
theory; so keeping the Lagrange multiplier $b$ instead of eliminating
it affords a certain universality to the formulation. The
corresponding conserved potential quasi-vorticity $\tilde{q}$ and
quasi-helicity $\tilde{H}$ are defined by
\begin{equation}
\tilde{q}:= \rho^{-1}\MM{\tilde{\omega}}\cdot \nabla {s}
\quad\hbox{with}\quad
\tilde{H} 
:= \int_\mathcal{D} 
\MM{\tilde{u}} \cdot \diff{\MM{x}} \wedge
d(\MM{\tilde{u}} \cdot \diff{\MM{x}})
= \int_\mathcal{D} \MM{\tilde{\omega}} 
\cdot\,\curl^{-1}\MM{\tilde{\omega}} \diff{V}
\,,
\label{quasi-PV-helicity}
\end{equation}
in terms of the quasi-vorticity $\MM{\tilde{\omega}}$ in (\ref{quasi-vort2form}) and writing $\tilde{q}$ for an advected scalar function, $a_1={s}$, and mass density, $a_2=\rho\diff{V}$. 

\begin{example}\rm
As an example, consider the \emph{particular} case $a_1={s}$ and $a_2=\rho\diff{V}$, when an analogue of potential vorticity exists. In this case, the quasi-vorticity 2-form is given by 
\[
\MM{\tilde{\omega}} \cdot \diff{\MM{S}}  
=
\diff \left( (\MM{u}+\rho^{-1}b\,\nabla{s})\cdot \diff{\MM{x}} \right)
=
\curl\left(\MM{u}+\rho^{-1}b\,\nabla{s}\right)\cdot \diff{\MM{S}}
\,,
\]
and the Lagrange multiplier $b$ satisfies
\begin{equation*}
\left(\pp{}{t} +  \MM{u}\cdot\nabla \right)\frac{b}{\rho}
=
\frac{1}{\rho}\dede{l}{s} 
\,.
\label{b-eqn-pqv}
\end{equation*}
\end{example}

Thus, at the cost of keeping the $b$-equation as $\left(\pp{}{t} +
  \mathcal{L}_u\right)b={\delta l}/\delta{a}$, we can extend the
Kelvin, Ertel and helicity theorems for most fluid theories, but these
are Noether conservation laws for the Lagrangian corresponding to the extended
variable set $(u,a,b)$. 

\section{Conclusions} \label{sec-conclude} In this paper we showed how
to obtain conserved quantities for ideal fluid models that can be
obtained from the Euler-Poincar\'e equations with advected
quantities. The conserved quantities are obtained \emph{via} Noether's
Theorem from relabelling symmetries of the Lagrangian, which are
generated by vector fields 
that satisfy the condition 
$\eta_t-\ad_u\eta=0$.
Fluid theories usually involve advected
quantities that evolve according to the equation
$a_t+\mathcal{L}_ua=0$: an advected density is almost always present,
and other possibilities include advected scalars such as buoyancy, 
or advected 2-forms such as magnetic flux. In order to
have a symmetry of the Lagrangian it is also necessary to satisfy
equation \eqref{e:eta_advected_condition} so that the vector field $\eta$ generates a
symmetry of the advected quantity $a(t)$. In Corollary
\ref{Persistence-cor} we showed that if
\eqref{e:eta_advected_condition} is satisfied initially then it is
satisfied for all subsequent times. In general, defining a
parameterisation of the null space of $\mathcal{L}_\eta$ for $\eta$
then leads to evolution equations defined on the dual of the space of
parameterising functions. Note that this approach is different to that
of Section 4 of \cite{HoMaRa1998}, which does not use Noether's
Theorem and instead performs computations on the Euler-Poincar\'e
equations directly.

In this paper we considered fluid theories with advected density,
leading to conservation of vorticity, as well as advected density plus
advected tracers, leading to conservation of potential vorticity, and
advected density plus advected 1-forms (such as magnetic flux),
leading to global conservation of cross-helicity. There are many other
advected quantities that could be considered, notably advected tensor
fields which are used in the theory of ideal complex fluids
\cite{Ho2002}. However, since most fluid theories include an advected
density (even incompressible flow, for which the pressure is a
Lagrange multiplier that enforces that the density remain constant),
it is not always possible to simultaneously solve all of the
constraints on $\eta$ arising from the requirement that $\eta$
generates symmetries of all of the advected quantities involved. For
example, when a density and a magnetic flux are both present, there is
only one symmetry and hence only one globally conserved quantity.  An
interesting class of problems in which density is not necessarily
present arise in computational anatomy \cite{HoRaTrYo2004}. Here the
aim is to find the solution of the EPDiff equation (for which the
reduced Lagrangian is a functional of u only) which transports one
configuration of an advected quantity $a$ to another. The advected
quantity might be a scalar (for greyscale images), a singular measure
(for curves and surfaces), or even a tensor field (for diffusion
tensor images). Solutions of the EPDiff equation are geodesics on the
diffeomorphism group; for these solutions to drop to geodesics on the
shape space corresponding to the chosen advected quantities, all of
the conserved Noether quantities must vanish \cite{GaBaRa2011}. For
example, for greyscale images described by advected functions
$I(x,t)$, the momentum is constrained to be normal to the image:
$\delta{l}/\delta{u}\cdot \nabla l = 0.$ Hence, the Noether quantities
define the geometry of the shape space.

One of the ``holy grails'' in the field of variational numerical
methods is to find Eulerian discretisations of fluid dynamics that
arise from a variational principle. Amongst other things, this would
provide the possibility of discrete forms of the Noether's Theorem
described in this paper. One direction that we have previously
explored is to try to find a discretisation of the diffeomorphism
group and to obtain some form of reduction by symmetry; this approach
has been developed in some detail, making extensive use of discrete
exterior calculus, for the case of incompressible flows in
\cite{Pa+2011}.  In \cite{CoHo2009a}, it was shown that the spatial
discretisation of the Lie bracket must satisfy the closure property if
a reduction is to be obtained; it was also shown how space-time
discretisations could be obtained in this case. Another direction that
we have explored is using Clebsch constraints to enforce the evolution
of the back-to-labels map \cite{CoHoHy2007}. This leads to a
multisymplectic formulation of fluid dynamics that can be discretised
by a standard recipe but reduction (elimination of the back-to-labels
map) is not possible after discretisation due to symmetry breaking.

On the Hamiltonian side, the conserved quantities associated via
Noether's theorem with relabelling symmetry comprise the Casimir
functions. The variational derivatives of the Casimir functions are
null eigenvectors of the Lie-Poisson Hamiltonian structure that arises
from the Euler-Poincar\'e framework upon Legendre transforming. This
is explained further in \cite{HoMaRaWe1985,Ho1986,AbHo1987}. See also
\cite{PaMo1996a,PaMo1996b} for related discussions and additional
references. The conservation laws that follow from Noether's theorem
for relabelling symmetry of the Eulerian fluid variables generate
steady flows when substituted into the augmented Lie-Poisson brackets
that include the particle labels as functions of time and spatial
coordinate \cite{Ho1986,AbHo1987}.

Thus, Lie-Poisson brackets with the Casimir functions leave the
Eulerian fluid variables invariant, but they shift the fluid particle
labels along steady flows. This fact led to a strategy for proving
nonlinear stability of equilibrium flows that was first recognized in
\cite{Ar1965} for ideal incompressible planar flows and was later
developed and applied more widely in plasma physics in
\cite{HoMaRaWe1985}. Likewise, the Eulerian vector fields for
relabelling symmetries found here could just as well have been
obtained by solving for the null eigenvectors of the Lie-Poisson
bracket on the Hamiltonian side. Critical points of the sum of
Hamiltonian and the Casimirs are steady equilibrium flows. The
stability of these equilibrium flows may be studied by taking a second
variation and determining the conditions on the equilibrium that would
make the corresponding linearly conserved second variation sign
definite \cite{HoMaRaWe1985}.

\begin{remark}[Noether's other theorem]
  As mentioned earlier, Noether's original paper actually contains two
  theorems. The second one is generally regarded as the more subtle of
  the two. Noether's second theorem leads in principle to dependence
  among the Euler-Lagrange equations (Bianchi identities). However,
  for ideal fluids, we have not found any strictly Eulerian
  conservation laws in addition those already discussed here by using
  Noether's second theorem in the Euler-Poincar\'e context.
\end{remark}

\subsection*{Acknowledgments}
This paper was inspired by remarks made to CJC by Oliver B\"uhler
about Kelvin's theorem and the \cite{HoMaRa1998} paper. The authors
are also grateful to Y. Kosmann-Schwartzbach and E. L. Mansfield for
comments, encouragement and advice while we were writing of this paper,  
and the useful suggestions and corrections from the two anonymous reviewers.  
The work by DDH was partially supported by an Advanced Grant from the European Research Council. 

\bibliography{kelvin}

\end{document}